\documentclass[manuscript,screen]{acmart}

\usepackage{enumitem}
\usepackage{float}
\usepackage{graphicx}
\usepackage{color}
\usepackage{algorithm}
\usepackage{algorithmic}
\usepackage{diagbox}
\usepackage{multirow}
\usepackage{subfigure}

\AtBeginDocument{%
  }

\setcopyright{acmcopyright}
\begin{document}

\title{Two new algorithms for solving  M\"uller games and their applications}

\author{Zihui Liang}
\authornote{Corresponding authors.}
\email{zihuiliang.tcs@gmail.com}
\orcid{0000-0002-9022-6470}
\author{Bakh Khoussainov}
\authornotemark[1]
\email{bmk@uestc.edu.cn}
\author{Mingyu Xiao}
\email{myxiao@uestc.edu.cn}
\orcid{0000-0002-1012-2373}
\affiliation{%
  \institution{University of Electronic Science and Technology of China}
  \streetaddress{2006 Xiyuan Avenue}
  \city{Chengdu}
  \state{Sichuan}
  \country{China}
  \postcode{611731}
}

\renewcommand{\shortauthors}{Liang et al.}

\begin{abstract}
M\"uller games form a well-established class of games for model checking and verification. These games are played on directed graphs $\mathcal G$ where Player 0 and Player 1 play by generating an infinite path through the graph. The winner is determined by  the set $X$ consisting of all vertices in the path that occur infinitely often. If $X$ belongs to $\Omega$, a specified collection  of subsets  of $\mathcal G$, then Player 0 wins. Otherwise, Player 1 claims the win. These games are determined, enabling  the partitioning of $\mathcal G$ into two sets $W_0$ and $W_1$ of winning positions for Player 0 and Player 1, respectively.   Numerous algorithms exist that decide Müller games $\mathcal G$ 
by computing the sets $W_0$ and $W_1$. 
In this paper, we introduce two novel algorithms that outperform all previously known methods for deciding explicitly given Müller games, especially in the worst-case scenarios. The previously known algorithms either reduce M\"uller games to other known games (e.g. safety games) or recursively change the underlying graph $\mathcal G$ and the collection of sets in $\Omega$.  In contrast, our approach does not employ these techniques but instead leverages subgames, the sets within $\Omega$, and their interactions. This distinct methodology sets our algorithms apart from prior approaches for deciding 
Müller games. Additionally, our algorithms offer enhanced clarity and ease of comprehension.   Importantly, our techniques are applicable not only to M\"uller games but also to improving the performance of existing algorithms that handle other game classes, including coloured M\"uller games, McNaughton games,  Rabin games, and Streett games. 
\end{abstract}


\maketitle

\section{Introduction}

In the area of verification and synthesis of reactive systems, model checking, and logic, studying games played on finite graphs is a key research topic \cite{gradel2002automata}. The most current work \cite{fijalkow2023games} serves as an excellent reference for the state-of-the-art methods in this area. Interest in these games primarily arises from their role in modeling and verifying reactive systems and their specifications as games on graphs. These games are played on finite directed graphs between Player 0 (the controller) and Player 1 (the adversary, e.g., the environment). 
The players engage in ongoing interactions with each other, and the winner is determined by the long-term behavior of the players.  M\"uller games, McNaughton games, coloured M\"uller games, Rabin games, and Streett games constitute  well-established classes of games for verification.  These games are played on bipartite
graphs $\mathcal G$  where Player 0 and Player 1 play the game by producing an infinite path $\rho$ in $\mathcal G$.  Then the winner of
this play is determined by conditions put on $\mathsf{Inf}(\rho)$ the set  of all vertices in the path that appear infinitely often. Thus, 
the winning conditions depend solely on those vertices 
that occur infinitely often in the given play $\rho$. 
Understanding the algorithmic content of determinacy results for these games  is at the core of the area.

All games that we listed above, including M\"uller games, are played in arenas that we define below:

\begin{definition}
    An {\bf arena} $\mathcal{A}$, or equivalently a {\bf game graph}, is a bipartite  directed graph  $(V_0, V_1, E)$, where 
\begin{enumerate}
    \item   $V_0\cap V_1=\emptyset$, and $V=V_0\cup V_1$ is the set of nodes of  $\mathcal A$. The nodes of $V$ will also be called {\bf positions}.
    \item $E\subseteq V_0 \times V_1 \cup V_1 \times V_0$ is the set of edges such that every node has an outgoing edge. 
    \item $V_0$ and $V_1$ are sets of positions from which Player 0 and Player 1, respectively, move. Nodes in $V_0$ are called Player 0 positions, and nodes in $V_1$ are Player 1 positions.
\end{enumerate}
\end{definition}

Let $\mathcal A$ be an arena. Players play the game in the arena $\mathcal A$ by taking turns and moving a token along the edges of the underlying graph. Initially, the token is placed on a node $v_0 \in V$. 
If $v_0\in V_0$, then Player 0 moves first. Conversely, if $v_0\in V_1$, then Player 1 moves first.   
In each round of play, if the token is positioned 
on a Player $\sigma$'s position $v$, then Player $\sigma$ chooses $u\in E(v)$, moves the token to $u$ along the edge $(v,u)$, and the play continues on to the next round. Formally, 

\begin{definition}
Let $\mathcal A$ be an arena. A {\bf play}, that starts at position $v_0$, is an infinite sequence $\rho=v_0,v_1, v_2, \ldots$ such that $v_{i+1}\in E(v_i)$ for all $i\in \mathbb{N}$.  \ Note that in the play we used  the assumption $E(v)\neq \emptyset$.
\end{definition}

Given a play  $\rho=v_0, v_1, \ldots$, the set  
$\mathsf{Inf}(\rho)=\{v\in V \mid \exists^{\omega} i (v_i=v)\}$ 
is called the {\bf infinity set} of $\rho$. The winner of this play  is determined by  a condition put on $\mathsf{Inf}(\rho)$. We list several of these conditions that are well-established in the area.

\begin{definition} \label{dfn:WinCon}
Let $\mathcal{A}=(V_0,V_1,E)$ be an arena. All the games below are called {\bf regular} games: 
\begin{enumerate}
\item A {\bf M\"uller game} is the tuple $\mathcal{G}=(\mathcal{A}, \Omega)$, where  $\Omega \subseteq 2^{V}$.  Sets in $\Omega$ are called {\bf winning conditions}. We say that Player 0 {\bf wins} the play $\rho=v_0, v_1, \ldots$
if $\mathsf{Inf}(\rho) \in \Omega$. Otherwise, Player 1 wins. 
\item A {\bf McNaughton game} is the tuple $\mathcal{G}=(\mathcal{A}, W, \Omega)$, where 
$W\subseteq V$, and $\Omega \subseteq 2^{W}$ is a collection of {\bf 
winning conditions}. Player 0 {\bf wins} $\rho=v_0,v_1,\ldots$ if $\mathsf{Inf}(\rho)\cap W\in \Omega$. Else, Player 1 wins.
\item 
A {\bf coloured M\"uller game} is $\mathcal{G}=(\mathcal{A}, c, \Omega)$,  where $c: V\rightarrow C$ is a mapping from $V$ into the set $C$ of colors, and $\Omega \subseteq 2^{C}$. \  Call sets in $\Omega$ {\bf winning conditions}. Player 0 {\bf wins} $\rho=v_0, v_1, \ldots$
if $c( \mathsf{Inf}(\rho)) \in \Omega$. Else, Player 1 wins. 

\item A {\bf Rabin game} is the tuple $\mathcal{G}=(\mathcal{A},(U_1, V_1), \ldots, (U_k, V_k))$, where 
$U_i, V_i \subseteq V$, $(U_i, V_i)$ is a {\bf winning condition}, and the {\bf index $k\geq 0$} is an integer. Player 0 {\bf wins} $\rho=v_0, v_1, \ldots$ if there is a pair $(U_i, V_i)$ such that 
$\mathsf{Inf}(\rho) \cap U_i\neq \emptyset$ and $\mathsf{Inf}(\rho) \cap V_i= \emptyset$. Else, Player 1 wins.
\item A {\bf Streett game} is the tuple $\mathcal{G}=(\mathcal{A},(U_1, V_1), \ldots, (U_k, V_k))$, where 
$U_i$, $V_i$ are as in Rabin game. Player 0 {\bf wins} $\rho=v_0, v_1, \ldots$ if for all $i\in \{1, \ldots, k\}$
if $\mathsf{Inf}(\rho) \cap U_i\neq \emptyset$ then $\mathsf{Inf}(\rho) \cap V_i\neq \emptyset$. Otherwise, Player 1 wins.

\item A {\bf KL game} is the tuple $\mathcal{G}=(\mathcal{A}, (u_1, S_1),\ldots,(u_t,S_t))$, where $u_i\in V$, $S_i\subseteq V$ is a {\bf winning condition}, and the {\bf index $t\geq 0$} is an integer. Player 0 {\bf wins} $\rho=v_0, v_1, \ldots$ if there is a pair $(u_i, S_i)$ such that 
$u_i\in \mathsf{Inf}(\rho)$ and $\mathsf{Inf}(\rho) \subseteq S_i$. Else, Player 1 wins. 
\end{enumerate} 
\end{definition}

The first three games are symmetric, e.g., for the M\"uller game $\mathcal{G}=(\mathcal{A}, \Omega)$ its  symmetric counter-part $(\mathcal{A}, 2^V\setminus \Omega)$ is also M\"uller game. Player 0 loses in game $(\mathcal{A}, \Omega)$ if and only if Player 1 wins in  $(\mathcal{A}, 2^V\setminus \Omega)$. Rabin games can be considered as Streett games. Player 0 wins  Rabin game $\mathcal G$ if and only if Player 0 loses the Streett game $\mathcal G$. The first five winning conditions have become well-established. The last condition is new. The motivation behind this new winning condition lies in the transformation of Rabin and Streett games into Müller games via the KL winning condition. In a precise sense, as will be seen in Section \ref{S:Applications}, the KL condition serves as a compressed Rabin winning condition.

The games defined above possess natural {\bf parameters}. In M\"uller game, the parameter is $\Omega$. In McNaughton games the parameter is the pair $(W, \Omega)$. In colored M\"uller games the parameter is $(C, \Omega)$. For $KL$ games the parameter is $(u_1, S_1),\ldots,(u_t,S_t))$ and the index $t$. In Rabin and Streett games the parameter is the winning condition sequence $(U_1, V_1), \ldots, (U_k, V_k)$ and the index $k$, the length of the sequence of the winning condition pairs. We denote the parameter values by $p$, so the value of $p$ belong to the set $\{|W|, |C|, t, k, |\Omega|\}$. 
The values of these parameters, when $p=t$, $p=k$ or $p=|\Omega|$, can be exponential on the size of the arenas. 


\begin{definition}
Let $\mathcal G$ be any of the regular games above. We say that $\mathcal G$ is {\bf explicitly given} if $V$, $E$, and all the winning conditions of the game $\mathcal G$, e.g., the sets in $\Omega$
in case $\mathcal G$ is a M\"uller game, are fully presented as input. 
\end{definition}

For instance, the (input) size of explicitly given M\"uller game is thus bounded by $|V|+|E|+ 2^{|V|}\cdot |V|$. 
In particular, the explicit representation of any of the regular games can be exponential on the size of the arena of the game. 

A strategy for Player $\sigma$ is a function that receives as input initial segments of plays $v_0,v_1,\ldots, v_k$ where $v_k\in V_\sigma$ and outputs some $v_{k+1}$ such that $v_{k+1}\in E(v_k)$. For regular games, an important class of strategies are finite state strategies.  
The key is that these strategies depend only on a finite bounded part of the full history of the plays. R. McNaughton  in \cite{mcnaughton1993infinite} proved that the winner in McNaughton games always has a finite state winning strategy.  W. Zielonka proves 
that the winners of regular games have finite state winning strategies \cite{zielonka1998infinite}.  


In the study of regular games, the focus is naturally placed on  solving  them. Solving a given regular game entails two key objectives. First, one aims to devise an algorithm that, when provided with a regular game $\mathcal G$, partitions the set $V$ of positions into two sets $Win_0$ and $Win_1$ such that $v\in Win_{\sigma}$ if and only if Player $\sigma$ wins the game starting at $v$, where $\sigma \in \{0,1\}$. We call this the {\bf decision problem} where one wants to find out the winner of the game. Second,  one would like to design an algorithm that, given a regular game, extracts a winning strategy for the victorious player. This is known as the {\bf synthesis problem} where one wants to design a winning strategy for the winner.  





\section{Background and our contribution}

In this section, we briefly provide a background on algorithms that solve the above mentioned games with an emphasis on M\"uller games. We then introduce the basic well-established concepts needed in the study of games played on graphs. 

\subsection{Known algorithms} \label{S:Known-algorithms}

We start with M\"uller games. 
McNaughton in \cite{mcnaughton1993infinite} decides M\"uller games in time $O(a^{|V|}|V|! |V|^{3})$ for some constant $a>1$. He  
proves that the winner has a finite state winning strategy with at most $|V|!$ states.  Nerode, Remmel, and Yakhnis 
\cite{nerode1996mcnaughton} decide M\"uller games in $O(|V|!\cdot2^{|V|}|V|^3|E|)$. W. Zielenka \cite{zielonka1998infinite} examines M\"uller games  through specifically constructed Zielenka trees. The size of each Zielonka tree is $O(2^{|V|})$ in the worst case. \ S. Dziembowski, M. Jurdzinski, and I. Walukiewicz in \cite{dziembowski1997much} show that deciding M\"uller games with Zielonka trees as part of the input is in $\text{NP}\cap \text{co-NP}$. They also show that the bound $|V|!$ on the memory of winning strategies is sharp. 
D. Neider, R. Rabinovich, and M. Zimmermann reduce M\"uller games to safety games with $O((|V|!)^3)$ vertices and safety games can be solved in linear time \cite{neider2014down}.  F. Horn in \cite{horn2008explicit} provides the first polynomial time decision algorithm for explicitly given M\"uller games. The running time of his algorithm is $O(|V|\cdot |\Omega| \cdot (|V|+|\Omega|)^2)$.   F. Horn's correctness proof has a non-trivial flaw.  B. Khoussainov, Z. Liang, and M. Xiao in \cite{liang_et_al:LIPIcs.ESA.2023.79} provide a correct proof of Horn's algorithm through new techniques and methods.  Those techniques improve the running time of deciding M\"uller games to $O( |\Omega|\cdot(|V|+|\Omega|)\cdot |V_0|\log |V_0|  )$.

All the known algorithms that decide M\"uller games are either recursive algorithms or reductions to other known classes of games. Some recursive algorithms are based on induction techniques that decrease the sizes of arenas $\mathcal G$ or the winning condition $\Omega$, and then recompute the winning sets repeatedly. For instance, McNaughton algorithm, Nerode, Remmel, Yakhnis algorithm, and Zielenka's algorithm are of this type. These algorithms typically produce $|V|!$ running time for deciding M\"uller games. Other recursive algorithms are based on changing the structure of the underlying graphs and the winning sets. For instance, Horn's algorithm increases the size of the underlying set  to $|V|+|E|+|\Omega|+ |\Omega||V|$. An example of an algorithm that reduces M\"uller games to another class of known games is by D. Neider, R. Rabinovich, and M. Zimmermann \cite{neider2014down}. They reduce M\"uller games to safety games. Their reduction increases the size of the graph of the safety game to  $O((|V|!)^3)$. As we noted above, Horn's algorithm runs in time $O(|V|\cdot |\Omega| \cdot (|V|+|\Omega|)^2)$ polynomial on the size of the explicitly given M\"uller games. The degree of $|\Omega|$ in this bound is $|\Omega|^3$.  The degree of $|\Omega|$ in 
the bound $O( |\Omega|\cdot(|V|+|\Omega|)\cdot |V_0|\log |V_0|  )$ from \cite{liang_et_al:LIPIcs.ESA.2023.79} is $|\Omega|^2$. This is a significant reduction because the size of $\Omega$ can be exponential on $|V|$. 

With respect to McNaughton games, 
McNaughton \cite{mcnaughton1993infinite} provided the first algorithm that decides the games in time $O(a^{|W|}\cdot |W|! \cdot |V|^3)$, for a constant $a>1$. Nerode, Remmel, and Yakhnis in \cite{nerode1996mcnaughton} improved the bound to $O(|W||E||W|!)$. 
A. Dawar and P. Hunter proved that
finding the winner in McNaughton games is PSPACE-complete problem \cite{hunter2008complexity}. This implied that deciding the winner in games with all other winning conditions from Definition \ref{dfn:WinCon} is also PSPACE-complete \cite{hunter2008complexity}.  

As McNaughton games can easily be transformed into coloured M\"uller games, there has been a lot of work on designing algorithms for coloured M\"uller games. The standard algorithm that decides coloured M\"uller games uses induction on cardinality of $C$ \cite{fijalkow2023games}. These algorithms run in time $O(|C||E|(|C||V|)^{|C|-1})$.  C. Calude, S. Jain, B. Khoussainov, W. Li, and F. Stephan, using their breakthrough quasi-polynomial time algorithm for parity games, improve all the known algorithms for colored M\"uller games \cite{calude2017deciding}. Their algorithm runs in time $O(|C|^{5|C|}\cdot |V|^5)$. Bj\"orklund, Sandberg and Vorobyov \cite{bjorklund2003fixed} showed that under the Exponential Time Hypothesis it is impossible to decide colored M\"uller games in $O(2^{o(|C|)} \cdot |V|^a)$ 
for any constant $a$. C. Calude, S. Jain, B. Khoussainov, W. Li, and F. Stephan in \cite{calude2017deciding} improved this by showing that  under the Exponential Time Hypothesis it is impossible to decide colored M\"uller games in $2^{o(|C|\cdot \log(|C|))}Poly(|V|)$, where $|C| \leq \sqrt{|V|}$. 

We mention two algorithms with the best running time bounds for deciding Rabin and Streett games. Horn's algorithm  
for deciding Streett games has the running time  $O(k!|V|^{2k})$. N. Piterman and A. Pnuelli in 2006 provide algorithms that decide Rabin games $O(|E| |V|^{k+1}kk!)$ and Streett games in $O(nkk!)$ \cite{piterman2006faster}. 

Finally, the common feature of all these algorithms and their analysis is that they all take into account the {\bf parameters}.  
Hence they appear in describing the running times. We stress that the running times of these algorithms, in terms of the parameters $p$, where $p$ is either $|W|$ or $k$ or $|C|$, contain multiplicative term $p!$ or $p^p$.  Therefore, when the sizes of the parameters are large, all the algorithms mentioned above produce the worst case running times. In practice, these algorithms have limited power as they can be applied to games with rather small parameters.  In this paper, we design algorithms with exponential running time on the size of the vertex sets, thus outperforming known algorithms that decide 
games with large parameters. 


\subsection{Basic concepts} \label{S:basics}

To explain our contributions, we define standard well-established concepts used throughout our algorithms. 

\begin{definition}
    A {\bf pseudo-arena} of $\mathcal{A}$ determined by $X$ is $\mathcal{A}(X)=(X_0,X_1,E_X)$ where $X_0=V_0\cap X$, $X_1=V_1\cap X$, $E_X=E\cap (X\times X)$. If this pseudo-arena is an arena, then we call it the {\bf subarena} of $\mathcal{A}$ determined by $X$. By $\mathcal G(X)$ denote the M\"uller game played on the subarena $\mathcal{A}(X)$.
\end{definition}

Let us consider Player $\sigma$, where  $\sigma \in \{0,1\}$. The opponent of Player $\sigma$ is denoted by Player $\bar{\sigma}$. Traps are sub-arenas in games where one of the players has no choice but stay. Here is a formal definition: 
 
\begin{definition}[$\sigma$-trap]
  A subarena $\mathcal{A}(X)$ is a {\bf $\sigma$-trap} for Player $\sigma$ if each of the following two conditions are satisfied:  (1)  For all $x\in X_{\bar\sigma}$ there is a  $y\in X_{\sigma}$ such that $(x,y) \in E$.  (2)  For all $x \in X_{\sigma}$ it is the case that $E(x)\subseteq X$.
  \end{definition}
  Thus, if $\mathcal{A}(X)$ is a $\sigma$-trap, then Player $\bar{\sigma}$ can stay in $\mathcal{A}(X)$ forever 
  if the player wishes to do so. 
  



Let $T$ be a subset of the arena  $\mathcal{A}=(V_0,V_1,E)$. The attractor of Player $\sigma$ to the set $T\subseteq V$, denoted $Attr_\sigma(T,\mathcal{A})$, is the set of positions from where Player $\sigma$ can force the plays into $T$. The attractor $Attr_\sigma(T,\mathcal{A})$ is computed as follows:

\smallskip

\noindent
$W_0=T$, \ \ 
$W_{i+1}=W_i\cup \{u\in V_\sigma\mid E(u)\cap W_i \ne \emptyset\} \cup \{u\in V_{\bar\sigma} \mid E(u)\subseteq W_i\}$, \ and then set \  
$Attr_\sigma(T,\mathcal{A})=\bigcup_{i\geq 0}W_i$. 

\smallskip

The set $Attr_\sigma(T,\mathcal{A})$ can be computed in $O(|E|)$. We call $Attr_\sigma$ the attractor operator. Note that the set $V \setminus Attr_\sigma(T,\mathcal{A})$, the complement of the $\sigma$-attractor of $T$, is a $\sigma$-trap for all $T$. This set is the emptyset if and only if $V=Attr_\sigma(T,\mathcal{A})$.

\smallskip

As explained in Section \ref{S:Known-algorithms}, all the previously known algorithms take into account the parameters $p\in \{|W|, |C|, k\}$, and their running times contain the multiplicative terms  
$p^p$ or $p!$. Thus, these algorithms are suited for games with small parameters,  and they are prohibitively slow when the game parameters are large.  Hence, designing exponential time algorithms for large games  are important as they greatly outperform all the known algorithms for regular games.

In order to address this large vs small games issue we utilize the Lambert $\boldsymbol{W}$ function $LW$.
The function $LW(z)$ is a solution to the equation $we^w=z$. By \cite{hoorfar2008inequalities}, for $z>1$, we have the following:
$$
LW(z)\ge \frac{\log z}{1+\log z}(\log z-\log\log z + 1)>\log\log z.
$$
It is easy to see the following sequence of implications that are derived from comparing $p^p$ and $e^n$, where $n=|V|$: \ 
$$
p^p> e^n \Rightarrow e^{p\ln p}>e^n \Rightarrow p\ln p > n \Rightarrow e^{\ln p}\ln p > n \Rightarrow \ln p > LW(n) 
\Rightarrow p>e^{LW(n)}.
$$

Since $LW(n)e^{LW(n)}=n$ we get the equality 
$ e^{LW(n)}=\frac{n}{LW(n)}$. This implies  $p>\frac{n}{LW(n)}$.

\begin{definition}
Let $p \in \{|C|,|W|, k\}$ be a game parameter.  A regular game $\mathcal G$ is {\bf large} if $p >\frac{c\cdot n}{LW(c\cdot n)}$.
\end{definition}

 It can formally be argued that with increasing $n$, 
 the probability of selecting a large game tends to $1$. For example, let us randomly select a McNaughton game $\mathcal G$ played on arena of size $n$. Each game $\mathcal{G}$ contains a unique ordered winning condition $\Omega$. Then the number of total games and the number of small games on the arena of $\mathcal G$ are:
 $$
 \sum_{i=0}^{2^n}\binom{2^n}{i}i! \ = \ \sum_{i=0}^{2^n}\frac{2^n!}{(2^n-i)!} \  \  \  \  \  \   \   \  \   \mbox{and} \  \  \  \  \  \ \  \    \  \  \  \  \  \sum_{i=0}^{ 2^{\lfloor\frac{n}{LW(n)}\rfloor}}\frac{2^n!}{(2^n-i)!}<\sum_{i=0}^{ 2^{\lfloor\frac{n}{\log \log n}\rfloor}}\frac{2^n!}{(2^n-i)!}, \ \mbox{respectively.}
 $$
 When $n\rightarrow \infty$, the ratio of the small games to large games approaches $0$. Hence, the probability of selecting a large game approaches to $1$. Therefore, understanding decision algorithms for large games is an important, and theoretically natural, issue. Also, this paper motivates the study of small games as for large games we provide efficient solutions.

\section{Our contribution}

We list our three main contributions:  
\begin{itemize}
\item  We develop two algorithms for deciding Müller games.
We start with  M\"uller games because (1) they can be decided in polynomial time when given explicitly, (2) they serve as a platform for demonstrating our core concepts and the data structure.  The first algorithm runs in time $O(3^{|V|}(|V|+|E|))$. The second algorithm runs in time $O(3^{|V|}|V|)$.  
By utilizing these two algorithms, we provide the most efficient polynomial time algorithms to date that decide explicitly given M\"uller games. 
To illustrate this, 
when M\"uller games are large,  the best known algorithm  runs in time $O( |\Omega|\cdot(|V|+|\Omega|)\cdot |V_0|\log |V_0|)$ \cite{liang_et_al:LIPIcs.ESA.2023.79}.  Our first algorithm runs in time $O(3^{|V|} (|V|+|E|))$ and the second  in  $O(3^{|V|} |V|)$. These are, obviously, important improvements. 

\item Our algorithms distinguish themselves 
from the previously known algorithms in three ways. First, our algorithms neither reduce the sizes of the arenas nor alter the winning conditions.  
This is the feature of the many recursive algorithms that decide M\"uller games. Second, our algorithms avoid the transformation of Müller games into other well-known classes of games, such as safety games. This is contrary to most reduction techniques employed in Müller game decision processes. Thus, in terms of these two aspects, 
our techniques are novel. Third,  our methods  are based on well-established notions such as subarena, traps, and the attractor operator that we already defined in Section \ref{S:basics}. Our algorithms interplay these notions making them clean and simple, and hence easy to implement.  The  central technical concept used in this interplay is the notion of {\em full win}. A player {\bf fully wins} a subarena $\mathcal A(X)$ if the player wins the M\"uller game $\mathcal G(X)$ from any position in $X$. Our algorithms collect all the subarenas that Player 0 fully wins, and then, based on this collection, decide $\mathcal G$.

\item Finally, our methods are universal in the following sense. We can apply our methods directly to decide all other regular games. This is an obvious advantage and  distinction of our algorithms from all the other algorithms that solve regular games. The running times of known algorithms that decide regular games have parameters in them. Using our methods for deciding 
M\"uller games, we show that McNaughton games and colored M\"uller games can be decided in time $O(3^{|V|}\cdot |V|)$. For Rabin (and Streett) games we have the running bound $O((3^{|V|}+k)\cdot |V|^2)$. 
With this, we significantly improve the running bounds of all the known algorithms  when games are large. As an example, we 
improve the known bound for coloured M\"uller games obtained from the breakthrough quasi-polynomial time algorithm from \cite{calude2017deciding}. 
\end{itemize}

The table below summarises our results and compares them to the state of the art.  

\smallskip

\begin{center}
\begin{tabular}{||c|c| c||} 
 \hline
 \  & Best known running times & Our algorithm (s) \\ [0.5ex] 
 \hline\hline
 M\"uller games  & $O( |\Omega|\cdot(|V|+|\Omega|)\cdot |V_0|\log |V_0|  )$ & $O(min\{ |\Omega|\cdot(|V|+|\Omega|)\cdot |V_0|\log |V_0|, \ 3^{|V|}|V|\})$  \\ 
 \hline
 McNaughton games &  $O(|W||E||W|!)$ & \multirow{2}{*}{$O(3^{|V|}|V|)$}   \\
 \cline{1-2}
 Colored M\"uller games & $O(|C|^{5|C|}\cdot |V|^5)$ &  \\
 \hline
 Rabin games & $O(|E| |V|^{k+1}kk!)$ & \multirow{2}{*}{$O((3^{|V|}+2^{|V|}k)\cdot |V|)$, $O((3^{|V|}+k)|V|^2)$} \\
 \cline{1-2}
 Streett games & $O(|V|kk!)$ &  \\ 
 \hline
  KL games & none & $O((3^{|V|}+2^{|V|}t)\cdot |V|)$, $O(3^{|V|}|V|^2)$  \\
 \hline
\end{tabular}
\end{center}

\medskip

Other important comment is this. We mentioned the result by C. Calude, S. Jain, B. Khoussainov, W. Li and F. Stephan stating that  under the Exponential Time Hypothesis 
coloured M\"uller games cannot be decided  in  $2^{o(|C|\cdot \log(|C|))}Poly(|V|)$, where $|C|\leq \sqrt{|V|}$\footnote{In their paper  \cite{calude2017deciding}, C. Calude, S. Jain, B. Khoussainov, W. Li and F. Stephan  claim a misleading statement that it is impossible to decide coloured M\"uller games in time $2^{o(|C|\cdot \log(|C|))}Poly(|V|)$,  where $|C|\leq |V|$.  
However, their proof actually implies that $|C|\leq \sqrt{|V|}$.}. 
Our algorithm shows that when $|C|> \frac{\ln3\cdot n}{LW(\ln 3\cdot n)}$ we can solve coloured M\"uller games most efficiently. Indeed, when the game is large, 
our algorithm runs in $O(3^{|V|}|V|)$ which is $2^{o(|C|\cdot \log(|C|))}Poly(|V|)$. Also, by the mentioned result of  Bj\"orklund, Sandberg and Vorobyov \cite{bjorklund2003fixed}, under the Exponential Time Hypothesis, our results are almost optimal for coloured M\"uller games.  We do not know if there is a better exponential time algorithm that decides coloured M\"uller games, where $|C|$ belongs to the interval 
$(\sqrt{|V|}, \ \frac{\ln3\cdot |V|}{LW(\ln 3\cdot |V|)})$.


\section{Deciding M\"uller Games}\label{S:SMG}

For this section let us fix a M\"uller game $\mathcal G$ played on arena $\mathcal A$. Let  $Win_\sigma(\mathcal{G})$ be the set of all $v$ in $\mathcal{G}$ such that player $\sigma$ wins $\mathcal{G}$ starting from $v$. An important notion will be the following:

\begin{definition}
If $Win_\sigma(\mathcal{G})=V$, then 
player $\sigma$ {\bf fully wins $\mathcal{G}$}. Otherwise, we say that player $\sigma$ {\bf cannot fully win}  $\mathcal{G}$.
\end{definition}
Note that even if Player $\sigma$ cannot fully win $\mathcal G$, there might still be positions $v$ that the player wins $\mathcal G$ starting at $v$. We would like to collect all the subarenas $X$ such that a given player fully wins the game played on $\mathcal G(X)$. Note that even if Player $\sigma$ fully wins $\mathcal G(X)$, this does not imply that $X$ is a $\bar{\sigma}$-trap. Now we start some analysis of subarenas. 

\begin{lemma}\label{L: exists cannot win, exists opponent wins}
    If there exists a $\sigma$-trap $\mathcal{A}(X)$ so that Player $\sigma$ cannot fully win $\mathcal{G}(X)$, then there exists a $\sigma$-trap  $\mathcal{A}(Y)
    $ so that $Y\subseteq X$ and Player $\bar{\sigma}$ fully wins $\mathcal{G}(Y)$.
\end{lemma}
\begin{proof}
   Let $\mathcal{A}(X)$ be a $\sigma$-trap that  Player $\sigma$ cannot fully win. Let $Y=Win_{\bar\sigma}(\mathcal{G}(X))$. Then $Y\ne \emptyset$, Player $\bar \sigma$ fully wins $\mathcal{G}(Y)$ and $\mathcal{A}(Y)$ is a $\sigma$-trap in both $\mathcal{G}(X)$ and $\mathcal{G}$.
\end{proof}

\begin{corollary}\label{C: all cannot win, opponent wins all}
     If no $\bar \sigma$-trap $\mathcal{A}(X)$ with $X\subsetneq V$ exists that Player $\sigma$ fully wins, then Player $\bar \sigma$ fully wins all these $\bar \sigma$-traps. 
\end{corollary}
\begin{proof}
    Assume that there is a $\bar\sigma$-trap $\mathcal{A}(Y)$ with $Y\subsetneq V$ such that Player $\bar\sigma$ cannot fully win $\mathcal{G}(Y)$. Then by Lemma \ref{L: exists cannot win, exists opponent wins}, there exists a $\bar\sigma$-trap $\mathcal{A}(Z)$ with $Z\subseteq Y$ so that Player $\sigma$ fully wins $\mathcal{G}(Z)$. This is a contradiction.
\end{proof}

\begin{lemma}\label{L: 0 no 1-trap}
    If $X\notin \Omega$ and there does not exist a 1-trap $\mathcal{A}(X)$ with $X\subsetneq V$ in $\mathcal{G}$, then Player 1 fully wins $\mathcal{G}$.
\end{lemma}
\begin{proof}
    Note that for all $v\in V$ we have 
    $Attr_1(\{v\},\mathcal{A})=V$. Otherwise, for some $v\in V$ we will have a 1-trap $\mathcal{A}(X)$ with $X=V \setminus Attr_1(\{v\},\mathcal{A})$.  Now we construct a winning strategy for Player 1 as follows. Let $v_0,v_1,\ldots,v_{k-1}$ be all positions in $\mathcal G$. Initially set $i=0$. Player 1 forces the token to $v_i$ and once the token arrives at $v_i$, set $i=i+1\mod k$. With this strategy, the token is moved through each position infinitely often. Since 
    $V\notin \Omega$, Player 1 fully wins $\mathcal{G}$.
\end{proof}

Now we characterise all subarenas $X$ that are fully won by the players. Of course, our characterization will be based on whether or not $X\in \Omega$. 

\begin{lemma}\label{L:0 X in Omega}
    Let  $X\in \Omega$ be a subarena in $\mathcal{G}$. Player 0 fully wins $\mathcal{G}(X)$ if and only if for all 0-traps $\mathcal{A}(Y)$ with $Y\subsetneq X$ in $\mathcal{G}(X)$, Player 0 fully wins $\mathcal{G}(Y)$.
\end{lemma}
\begin{proof}
    Assume that there exists a 0-trap $\mathcal{A}(Y)$ with $Y\subsetneq X$ in $\mathcal{G}(X)$ so that Player 0 cannot fully win $\mathcal{G}(Y)$. By Lemma \ref{L: exists cannot win, exists opponent wins}, there exists a 0-trap $\mathcal{A}(Z)$ in $\mathcal{G}(X)$ so that Player 1 fully wins $\mathcal{G}(Z)$. Therefore, since $Z$ is a $0$-trap, Player 1 wins the game $\mathcal G$   starting from any position $v\in Z$. Indeed, Player $1$ keeps the token inside $\mathcal{A}(Z)$, and follows the winning strategy in $\mathcal{G}(Z)$. Therefore, Player 0 cannot fully win $\mathcal{G}(X)$.

    Assume that for all 0-traps $\mathcal{A}(Y)$ with $Y\subsetneq X$ in $\mathcal{G}(X)$, Player 0 fully wins $\mathcal{G}(Y)$. We construct the following winning strategy for Player 0 in $\mathcal{G}(X)$. Let $X=\{v_0,v_1,\ldots,v_{k-1}\}$ and $i$ initially be 0.
    \begin{itemize}
        \item If the token is in $Attr_{0}(\{v_i\},\mathcal{A}(X))$, then Player 0 forces the token to $v_i$ and once the token arrives at $v_i$, sets $i=i+1 \mod k$. Otherwise,
        \item since $\mathcal{A}(X\setminus Attr_{0}(\{v_i\},\mathcal{A}(X)))$ is a 0-trap in $\mathcal{G}(X)$, Player 0 uses a winning strategy in  $\mathcal{G}(X\setminus Attr_{0}(\{v_i\},\mathcal{A}(X)))$. 
    \end{itemize}
Consider any play consistent with the strategy described.     If the token finally stays in $\mathcal{A}(X\setminus Attr_{0}(\{v_i\},\mathcal{A}(X)))$  for some $i$, then Player 0 wins the game. Otherwise, the token must be moved through every vertex in $X$ infinitely often. Since $X\in \Omega$,  Player 0 wins. This implies that Player 0 fully wins $\mathcal{G}(X)$. 
\end{proof}

\begin{corollary}\label{C:1 X not in Omega}
    Let  $X\in 2^{V}\setminus \Omega$ be a subarena in $\mathcal{G}$. Player 1 fully wins $\mathcal{G}(X)$ if and only if for all 1-traps $\mathcal{A}(Y)$ with $Y\subsetneq X$ in $\mathcal{G}(X)$, Player 1 fully wins $\mathcal{G}(Y)$.
\end{corollary}

\begin{proof}
The proof follows from the symmetry of M\"uller games, and
Lemma \ref{L:0 X in Omega} above. 
\end{proof}

Lemma \ref{L:0 X in Omega} considers the case when $X\in \Omega$ forms a subarena, and provides necessary and sufficient conditions for Player 0 to fully win the game $\mathcal G(X)$. The next lemma considers the case when $X$ is a subarena but $X$ is a winning condition for Player 1, that is, $X \in 2^V\setminus \Omega$. The lemma provides necessary and sufficient conditions for Player 0 to fully win the game $\mathcal G(X)$. 

\begin{lemma}\label{L: 0 X notin Omega}
    Let  $X\in 2^{V}\setminus \Omega$ be a subarena in $\mathcal{G}$. Player 0 fully wins $\mathcal{G}(X)$ if and only if there exists a 1-trap $\mathcal{A}(Y)$ with $Y\subsetneq X$ in $\mathcal{G}(X)$ such that the following two conditions are satisfied:
    \begin{enumerate}
    \item Player 0 fully wins $\mathcal{G}(Y)$, and 
    \item $Attr_0(Y, \mathcal{A}(X))=X$ or Player 0 fully wins $\mathcal{G}(X\setminus Attr_0(Y, \mathcal{A}(X)))$. 
    \end{enumerate}
\end{lemma}
\begin{proof}
    Assume that there exists a 1-trap $\mathcal{A}(Y)$ with $Y\subsetneq X$  in $\mathcal{G}(X)$ such that Player 0 fully wins $\mathcal{G}(Y)$ and $Attr_0(Y, \mathcal{A}(X))=X$.
    Now define the following strategy for Player 0.
    Starting at any position in $X$, first force the play into the set $Y$. As soon as the token is placed in $Y$, use the winning strategy to fully win $\mathcal{A}(Y)$. Since $Y$ is a 1-trap, Player 1 fully wins $\mathcal G(X)$.
    
    Now consider the next case, where we assume that there exists a 1-trap $\mathcal{A}(Y)$ with $Y\subsetneq X$  in $\mathcal{G}(X)$ such that  Player 0 fully wins both games: $\mathcal{G}(Y)$ and $\mathcal{G}(X\setminus Attr_0(Y, \mathcal{A}(X)))$.  Below we  construct the following winning strategy for Player 0 that guarantees that the player fully wins $\mathcal G(X)$.
    \begin{itemize}
        \item If the token is in $\mathcal{A}(Attr_0(Y, \mathcal{A}(X)))$, then Player 0 forces the token into $\mathcal{A}(Y)$ and then follows a winning strategy that fully wins $\mathcal{G}(Y)$.
        \item Otherwise, Player 0 follows the winning strategy in $\mathcal{G}(X\setminus Attr_0(Y, \mathcal{A}(X)))$.
    \end{itemize}
    Consider any play $\rho$ consistent with the strategy. If the token in the play is placed into $\mathcal{A}(Attr_0(Y, \mathcal{A}(X)))$, then Player 0 wins just like in the previous case. Otherwise, the token along this play will never move into $\mathcal{A}(Attr_0(Y, \mathcal{A}(X)))$. Since Player 0 follows a winning strategy in
    $\mathcal{G}(X\setminus Attr_0(Y, \mathcal{A}(X)))$, the play must be won by Player 0. Therefore, Player 0 fully wins the game in $\mathcal{G}(X)$.  So, this proves one direction of the lemma.

    Now we prove the other direction of the lemma. We assume that Player 0 fully wins the game $\mathcal G(X)$. We need to consider several cases.

    {\em Case 1:} Assume that there is no 1-trap $\mathcal{A}(Y)$ with $Y\subsetneq X$ in $\mathcal{G}(X)$.  By Lemma \ref{L: 0 no 1-trap}, Player 1 fully wins $\mathcal{G}(X)$. This obviously contradicts with our assumption. 

    {\em Case 2:} Assume that for all 1-traps $\mathcal{A}(Y)$ with $Y\subsetneq X$, Player 0 cannot fully win $\mathcal{G}(Y)$. Then by Corollary \ref{C: all cannot win, opponent wins all}, Player 1 fully wins all these $\mathcal{G}(Y)$, and by Corollary \ref{C:1 X not in Omega}, Player 1 fully wins $\mathcal{G}(X)$. 

    {\em Case 3:} Assume that for all 1-traps $\mathcal{A}(Y)$ with $Y\subsetneq X$  in $\mathcal{G}(X)$, if Player 0 fully wins $\mathcal{G}(Y)$ then $Attr_0(Y,\mathcal{A}(X))\ne X$ and Player 0 cannot fully win $\mathcal{G}(X\setminus Attr_0(Y, \mathcal{A}(X)))$. Let $\mathcal{A}(Y)$ be any of such 1-traps. Since $\mathcal{A}(X\setminus Attr_0(Y, \mathcal{A}(X)))$ is a 0-trap in $\mathcal{G}(X)$, by Lemma \ref{L: exists cannot win, exists opponent wins}, there exists a 0-trap $\mathcal{A}(Z)$ in $\mathcal{G}(X)$ so that Player 1 fully wins $\mathcal{G}(Z)$. By forcing the token in $\mathcal{A}(Z)$ and following the winning strategy in $\mathcal{G}(Z)$, Player 1 wins $\mathcal{G}(X)$ starting from any $v$ in $Z$. Therefore, Player 0 cannot fully win $\mathcal{G}(X)$.
\end{proof}

Let $\mathcal{G}=(\mathcal{A},\Omega)$ be a M\"uller game where $V=\{v_1,v_2,\ldots,v_{n}\}$. We assign a $n$-bit binary number $i$ to each non-empty pseudo-arena $\mathcal{A}(S_i)$ in $\mathcal{G}$ so that $S_i=\{v_j\mid \text{the $j$th bit of $i$ is 1}\}$. We partition all subgames $\mathcal{G}(S_i)$ into two sets $P=\{S_i\mid i\in [1,2^n-1] \text{ and Player 0 fully wins }\mathcal{G}(S_i)\}$ and $Q=\{S_i\mid i\in [1,2^n-1]\text{ and Player 0 cannot fully win }\mathcal{G}(S_i)\}$ with the following algorithm.

\begin{figure}[H]
    \centering 
    \scriptsize
    \begin{tabular}{l}
        \textbf{Input}: A M\"uller game $\mathcal{G}=(\mathcal{A},\Omega)$\\
        \textbf{Output}: The partitioned sets $P$ and $Q$.\\
        $P\leftarrow \emptyset$, $Q\leftarrow \emptyset$;\\
        \textbf{for} $i=1$ to $2^n-1$ \textbf{do}\\
        \hspace*{4mm}$S_i\leftarrow \{v_j\mid \text{the $j$th bit of $i$ is 1}\}$;\\ \hspace*{4mm}$is\_win=$\text{false};\\
        \hspace*{4mm}\textbf{if} $\mathcal{A}(S_i)$ is not an arena \textbf{then}\\
        \hspace*{8mm}\textbf{break};\\
        \hspace*{4mm}\textbf{end}\\
        \hspace*{4mm}\textbf{if} $S_i\in \Omega$ \textbf{then}\\
        \hspace*{8mm} $is\_win\leftarrow\text{true}$\\
        \hspace*{8mm}\textbf{for} $S_j\subsetneq S_i$ \textbf{do}\text{\hspace{70mm}$\longrightarrow$ Lemma \ref{L:0 X in Omega}}\\
        \hspace*{12mm}\textbf{if} $\mathcal{A}(S_j)$ is a 0-trap in $\mathcal{G}(S_i)$ \text{and} $S_j\in Q$ \textbf{then}\\
        \hspace*{16mm}$is\_win\leftarrow$\text{false};\\
        \hspace*{16mm}\textbf{break};\\
        \hspace*{12mm}\textbf{end}\\
        \hspace*{8mm}\textbf{end}\\
        \hspace*{4mm}\textbf{else}\\
        \hspace*{8mm}\textbf{for} $S_j\subsetneq S_i$ \textbf{do}\text{\hspace{70mm}$\longrightarrow$ Lemma \ref{L: 0 X notin Omega}}\\
        \hspace*{12mm}\textbf{if} $\mathcal{A}(S_j)$ is a 1-trap in $\mathcal{G}(S_i)$ \text{and} $S_j\in P$ \textbf{then}\\
        \hspace*{16mm}\textbf{if} $Attr_0(S_j, \mathcal{A}(S_i))=S_i$ or $S_i\setminus Attr_0(S_j, \mathcal{A}(S_i))\in P$ \textbf{then}\\
        \hspace*{20mm}$is\_win\leftarrow$\text{true};\\
        \hspace*{20mm}\textbf{break};\\
        \hspace*{16mm}\textbf{end}\\
        \hspace*{12mm}\textbf{end}\\
        \hspace*{8mm}\textbf{end}\\
        \hspace*{4mm}\textbf{end}\\
        \hspace*{4mm}\textbf{if} $is\_win=$\text{true} \textbf{then}\\
        \hspace*{8mm}$P\leftarrow P\cup \{S_i\}$;\\
        \hspace*{4mm}\textbf{else}\\
        \hspace*{8mm}$Q\leftarrow Q\cup \{S_i\}$;\\
        \hspace*{4mm}\textbf{end}\\
        \textbf{end}\\
        \textbf{return} $P$ and  $Q$
    \end{tabular}
    \caption{Algorithm 1 for partitioning subgames of a M\"uller game}
    \label{F:Partition Muller game}
\end{figure}

We now explain the algorithm.  
The algorithm, given a M\"uller game $\mathcal{G}$ as input, and returns the collections $P$ and $Q$:
\begin{itemize}
    \item $P=\{S_i\mid i\in [1,2^n-1] \text{ and Player 0 fully wins }\mathcal{G}(S_i)\}$, and 
    \item $Q=\{S_i\mid i\in [1,2^n-1]\text{ and Player 0 cannot fully win }\mathcal{G}(S_i)\}$.
\end{itemize}
At each iteration, the algorithm either keeps both $P$ and $Q$ intact or extends either $P$ or $Q$. According to the algorithm, if $\mathcal{A}(S_i)$ is not an arena, then $S_i$ is disregarded. If $\mathcal{A}(S_i)$ is an arena, then by using Lemmas \ref{L:0 X in Omega} and 
\ref{L: 0 X notin Omega}, we put $\mathcal{A}(S_i)$ either into $P$ or into $Q$. 
\begin{enumerate}
    \item If $S_i\in \Omega$, then:
        \begin{enumerate}
            \item If there exists a 0-trap $\mathcal{A}(S_j)$ in $\mathcal{G}(S_i)$ so that $S_j\in Q$ then $S_i$ is added to $Q$.
            \item Otherwise, $S_i$ is added to $P$.
        \end{enumerate}
        \item Otherwise:
        \begin{enumerate}
            \item If there exists a 1-trap $\mathcal{A}(S_j)$ in $\mathcal{G}(S_i)$ so that $Attr_0(S_j,\mathcal{A}(S_i))=S_i$ or $S_i\setminus Attr_0(S_j,\mathcal{A}(S_i)) \in P$ then $S_i$ is added to $P$.
            \item Otherwise, $S_i$ is added to $Q$.
        \end{enumerate}
    \end{enumerate}

\begin{lemma}\label{L: P Q}
    At the end of Algorithm 1, we have the following two equalities:
    \begin{itemize}
        \item  $P=\{S_i\mid i\in [1,2^n-1] \text{ and Player 0 fully wins }\mathcal{G}(S_i)\}$, and 
        \item $Q=\{S_i\mid i\in [1,2^n-1]\text{ and Player 0 cannot fully win }\mathcal{G}(S_i)\}$.
    \end{itemize}    
\end{lemma}

\begin{proof}
    If $i=1$ then $\mathcal{A}(S_1)$ isn't an arena, and hence $S_1$ is disregarded. 
    For $i=2,3,\ldots,2^n-1$, we want to show that at the end of $i$th iteration, (1) if Player 0 fully wins $\mathcal{G}(S_i)$ then $S_i$ is added to $P$, and (2) if Player 0 cannot fully win $\mathcal{G}(S_i)$ then $S_i$ is added to $Q$. Assume for all $j=1,2,\ldots,i-1$, (1) if Player 0 fully wins $\mathcal{G}(S_j)$ then $S_j$ is added to $P$, and (2) if Player 0 cannot fully win $\mathcal{G}(S_j)$ then $S_j$ is added to $Q$. If $\mathcal{A}(S_i)$ isn't an arena then $S_i$ is disregarded. 
    Otherwise:
    \begin{enumerate}
        \item If $S_i\in \Omega$, then by Lemma \ref{L:0 X in Omega}, Player 0 fully wins $\mathcal{G}(S_i)$ if and only if for all 0-traps $\mathcal{A}(S_j)$ in $\mathcal{G}(S_i)$, Player 0 fully wins $\mathcal{G}(S_j)$. Since for all these $S_j$, $j<i$, we have that if Player 0 fully wins $\mathcal{G}(S_j)$, then $S_j\in P$, otherwise $S_j\in Q$. Therefore, Player 0 fully wins $\mathcal{G}(S_i)$ if and only if for all 0-traps $\mathcal{A}(S_j)$ in $\mathcal{G}(S_i)$, $S_j\in P$. 
        \item If $S_i\notin \Omega$, then by Lemma \ref{L: 0 X notin Omega}, Player 0 fully wins $\mathcal{G}(S_i)$ if and only if there exists a 1-trap $\mathcal{A}(S_j)$ with $S_j\subsetneq S_i$ in $\mathcal{G}(S_i)$ so that (1) Player 0 fully wins $\mathcal{G}(S_j)$ and (2) $Attr_0(S_j, \mathcal{A}(S_i))=S_i$ or Player 0 fully wins $\mathcal{G}(S_k)$ where $S_k=S_i\setminus Attr_0(S_j,\mathcal{A}(S_i))$. Since for all these $S_j$ and $S_k$, $j<i$ and $k<i$, we have that if Player 0 fully wins $\mathcal{G}(S_j)$ (or $\mathcal{G}(S_k)$), then $S_j\in P$ (or $S_k\in P$), otherwise $S_j\in Q$ (or $S_k\in Q$). Therefore, Player 0 fully wins $\mathcal{G}(S_i)$ if and only if there exists a 1-trap $\mathcal{A}(S_j)$ with $S_j\subsetneq S_i$ in $\mathcal{G}(S_i)$ so that (1) $S_j\in P$ and (2) $Attr_0(S_j, \mathcal{A}(S_i))=S_i$ or $S_i\setminus Attr_0(S_j, \mathcal{A}(S_i)) \in P$. 
    \end{enumerate}
    By hypothesis, the proof is done.
\end{proof}

\begin{lemma}\label{L: 1-trap union}
    Let  $\mathcal{A}(X)$ and $\mathcal{A}(Y)$ be 1-traps. 
    If Player 0 fully wins $\mathcal{G}(X)$ and $\mathcal{G}(Y)$ then Player 0 fully wins $\mathcal{G}(X\cup Y)$. 
\end{lemma}
\begin{proof}
    We construct a winning strategy for Player 0 in $\mathcal{G}(X\cup Y)$ as follows. \  If the token is in $Attr_0(X, \mathcal{A}(X\cup Y))$, Player 0 forces the token into $X$ and once the token arrives at $X$, Player 0 follows the winning strategy in $\mathcal{G}(X)$. \ Otherwise, Player 0 follows the winning strategy in $\mathcal{G}(Y)$.
\end{proof}


\begin{lemma}\label{L: P to winning region}
    If for all $S_i\in P$, the arena $\mathcal{A}(S_i)$ isn't 1-trap in $\mathcal{G}$, then $Win_0(\mathcal{G})=\emptyset$ and $Win_1(\mathcal{G})=V$. Otherwise, let $\mathcal{A}(S_{max})$ be the maximal 1-trap in $\mathcal{G}$ so that $S_{max}\in P$. Then $Win_0(\mathcal{G})=S_{max}$ and $Win_1(\mathcal{G})=V\setminus S_{max}$.
\end{lemma}
\begin{proof}
    For the first part of the lemma, assume  that $Win_0(\mathcal{G})\ne\emptyset$. By Lemma \ref{L: P Q}, for all arenas $\mathcal{A}(X)$, $X\in P$ if and only if Player 0 fully wins $\mathcal{G}(X)$. Now note that $Win_0(\mathcal{G})\ne\emptyset$ is 1-trap such that
     Player 0 fully wins $\mathcal G (Win_0(\mathcal{G}))$. This contradicts with the assumption of the first part. 
For the second part, consider all 1-traps $X$ in $P$. \  Player 0 fully wins the games $\mathcal G(X)$ in each of these 1-traps by definition of $P$. By Lemma \ref{L: 1-trap union}, Player 0 fully wins the union of these 1-traps. Clearly, this union is $S_{max} \in P$. Consider $V\setminus S_{max}$. This set is a $0$-trap. Suppose Player 1 does not win 
$\mathcal G(V\setminus S_{max})$ fully. Then there exists 
a $0$-trap $Y$ in game $\mathcal G(V\setminus S_{max})$ such that Player 0 fully wins $\mathcal G(Y)$. For every Player 1 position in $y\in Y$ and outgoing edge $(y,x)$ we have either $x\in Y$ or $x\in S_{max}$. This implies $S_{max}\cup Y$ is 1-trap  such that Player 0 fully wins $\mathcal G(S_{max}\cup Y)$. So, $S_{max}\cup Y$ must be in $P$. This contradicts with the choice of $S_{max}$. 
\end{proof}

\section{Implementation} \label{S:Implementations}

In this section, we will introduce the data structure and, based on the data structure, provide two 
algorithms for deciding  M\"uller games.

\subsection{Algorithm 1}

Let $\mathcal{G}=(\mathcal{A},\Omega)$ be a M\"uller game where $V=\{v_1,v_2,\ldots,v_{n}\}$. We already assigned  $n$-bit binary numbers $i$ to non-empty pseudo-arenas $\mathcal{A}(S_i)$ in $\mathcal{G}$, where $S_i=\{v_j\mid \text{the $j$th bit of $i$ is 1}\}$. \ With this encoding, we can apply a binary tree to maintain any given 
collection of vertex sets $\mathcal{S}=\{S_{i_1},S_{i_2},\ldots,S_{i_k}\}$ so that insertions, deletions and queries to any of these sets takes $O(n)$ time, traversing all $\mathcal S$ takes time $O(2^n)$, 
and building the binary tree from $\mathcal{S}$ takes $O(kn)$. So, from now on, we apply the binary trees to maintain $\Omega$, $P$ and $Q$. Building the binary tree from $\Omega$ takes $O(2^nn)$ time.

\begin{lemma} \label{L: P and Q implement}
    There exists an algorithm that computes $P$ and $Q$ for a M\"uller game in time $O(3^{|V|}\cdot (|V|+|E|))$.
\end{lemma}

\begin{proof}
    We use the Algorithm 1 from Figure \ref{F:Partition Muller game}. We enumerate all $S_i$ so that $\mathcal{A}(S_i)$ is an arena. Since checking whether a pseudo-arena is an arena takes $O(|E|)$ time, this process takes $O(2^{|V|}\cdot |E|)$ time. Then, we enumerate all $S_j$ with $S_j\subsetneq S_i$ and there are $\sum_{k=1}^{|V|}\binom{|V|}{k}(2^{k}-1)<3^{|V|}$ such pairs of $S_i$ and $S_j$. By applying the binary trees, the enumeration takes $O(3^{|V|})$ time. If $S_i\in \Omega$ then verifying whether $\mathcal{A}(S_j)$ is a 0-trap in $\mathcal{G}(S_i)$ takes $O(|E|)$ time and checking whether $S_j$ is in $Q$ takes $O(|V|)$ time. If $S_i\notin\Omega$ then verifying whether $\mathcal{A}(S_j)$ is a 1-trap in $\mathcal{G}(S_i)$ takes $O(|E|)$ time, checking whether a vertex set is in $P$ takes $O(|V|)$ time and computing $Attr_0(S_j,\mathcal{A}(S_i))$ takes $O(|E|)$ time. Hence the operations on $S_j$ takes $O(|V|+|E|)$ time. This algorithm runs in $O(3^{|V|} \cdot (|V|+|E|))$ time.
\end{proof}

\begin{lemma}\label{L: P to winning region implement}
    
    Given M\"uller game $\mathcal{G}$ and $P$, there exists an algorithm which computes $Win_0(\mathcal{G})$ and $Win_1(\mathcal{G})$ in time 
    $O(2^{|V|}\cdot (|V|+|E|))$.
\end{lemma}

\begin{proof}
    By Lemma \ref{L: P to winning region}, we enumerate $S_i$ from $i=2^{|V|}-1$ to $i=1$. Since checking whether $S_i\in P$ and $\mathcal{A}(S_i)$ is a 1-trap  in $\mathcal{G}$ takes $O(|V|+|E|)$ time, it takes $O(2^{|V|}\cdot (|V|+|E|))$ time to find the first $S_i$ so that $S_i\in P$ and $\mathcal{A}(S_i)$ is a 1-trap  in $\mathcal{G}$. If such $S_i$ exists, then $Win_0=S_i$ and $Win_1=V\setminus S_i$, otherwise $Win_0=\emptyset$ and $Win_1=V$. This algorithm takes $O(2^{|V|}\cdot (|V|+|E|))$ time in total. 
\end{proof}

By Lemmas \ref{L: P and Q implement} and \ref{L: P to winning region implement}, the following theorem is proved.

\begin{theorem}\label{Thm: solving muller game}
    There exists an algorithm that, given a M\"uller game $\mathcal{G}$,  decides $\mathcal G$ 
    in time $O(3^{|V|}\cdot (|V|+|E|))$.
\end{theorem}

\subsection{Algorithm 2}

We want to improve Algorithm 1 by reducing the 
computation of the attractor operator. For this, we need to  strengthen Lemma \ref{L: 0 X notin Omega} that will be used
in our next algorithm.

\begin{lemma}\label{L:0 X notin Omega strengthened}
    Let  $X\in 2^{V}\setminus \Omega$ be a subarena in $\mathcal{G}$. Player 0 fully wins $\mathcal{G}(X)$ if and only if there exists a 1-trap $\mathcal{A}(Y)$ with $Y\subsetneq X$ in $\mathcal{G}(X)$ such that the following condition is satisfied:
    \begin{enumerate}
    \item Player 0 fully wins $\mathcal{G}(Y)$ and $|Y|=|X|-1$, or
    \item Player 0 fully wins $\mathcal{G}(Y)$, $Y=Attr_0(Y, \mathcal{A}(X))$ and Player 0 fully wins $\mathcal{G}(X\setminus Y)$. 
    \end{enumerate}
\end{lemma}

\begin{proof} First we show that Lemma \ref{L: 0 X notin Omega} implies this lemma.  Let  $X\in 2^{V}\setminus \Omega$ be a subarena in $\mathcal{G}$. Assume that exists a 1-trap $\mathcal{A}(Y)$ with $Y\subsetneq X$ in $\mathcal{G}(X)$ such that Player 0 fully wins $\mathcal{G}(Y)$ and $Attr_0(Y, \mathcal{A}(X))=X$.  Then it is easy to see that there exists a 1-trap $\mathcal{A}(Y')$ with $|Y'|=|X|-1$ in $\mathcal{G}(X)$ so that Player 0 fully wins $\mathcal{G}(Y')$.  \ If there  exists a 1-trap $\mathcal{A}(Y)$ with $Y\subsetneq X$ in $\mathcal{G}(X)$ such that Player 0 fully wins $\mathcal{G}(Y)$ and $\mathcal{G}(X\setminus Attr_0(Y, \mathcal{A}(X)))$, then we set $Y'=Attr_0(Y, \mathcal{A}(X))$. Thus we have that Player 0 fully wins $\mathcal{G}(Y')$ and $\mathcal{G}(X\setminus Y')$. 
    
   Now we show that conditions (1) and (2) of this lemma imply Lemma \ref{L: 0 X notin Omega}.  If there  exists a 1-trap $\mathcal{A}(Y)$ with $Y\subsetneq X$ in $\mathcal{G}(X)$ such that Player 0 fully wins $\mathcal{G}(Y)$ and $|Y|=|X|-1$, then $Attr_0(Y, \mathcal{A}(X))=X$. If there  exists a 1-trap $\mathcal{A}(Y)$ with $Y\subsetneq X$ in $\mathcal{G}(X)$ such that Player 0 fully wins $\mathcal{G}(Y)$, $Y=Attr_0(Y,\mathcal{A}(X))$ and Player 0 fully wins $\mathcal{G}(X\setminus Y)$, then Player 0 also fully wins $\mathcal{G}(X\setminus Attr_0(Y,\mathcal{A}(X)))$. Therefore, we have that there exists a 1-trap $\mathcal{A}(Y)$ with $Y\subsetneq X$ in $\mathcal{G}(X)$ such that the following two conditions are satisfied:  (a) Player 0 fully wins $\mathcal{G}(Y)$, and 
(b) $Attr_0(Y, \mathcal{A}(X))=X$ or Player 0 fully wins $\mathcal{G}(X\setminus Attr_0(Y, \mathcal{A}(X)))$. These two conditions are statements of Lemma \ref{L: 0 X notin Omega}. 
\end{proof}

Now we apply Lemma \ref{L:0 X notin Omega strengthened}
that changes Algorithm 1 as follows.  Run Algorithm 1 but 
replace the part of Algorithm 1 that corresponds to Lemma \ref{L: 0 X notin Omega} with the following code.

\begin{figure}[H]
    \centering 
    \scriptsize
    \begin{tabular}{l}
        \hspace*{8mm}\textbf{for} $S_j\subsetneq S_i$ \textbf{do}\text{\hspace{70mm}$\longrightarrow$ Lemma \ref{L:0 X notin Omega strengthened}}\\
        \hspace*{12mm}\textbf{if} $\mathcal{A}(S_j)$ is a 1-trap in $\mathcal{G}(S_i)$ \text{and} $S_j\in P$ \textbf{then}\\
        \hspace*{16mm}\textbf{if} $|S_j|=|S_i|-1$ \textbf{then}\\
        \hspace*{20mm}$is\_win\leftarrow$\text{true};\\
        \hspace*{20mm}\textbf{break};\\
        \hspace*{16mm}\textbf{end}\\
        \hspace*{16mm}\textbf{if} $S_j=Attr_0(S_j,\mathcal{A}(S_i)) $ and $S_i\setminus S_j\in P$ \textbf{then}\\
        \hspace*{20mm}$is\_win\leftarrow$\text{true};\\
        \hspace*{20mm}\textbf{break};\\
        \hspace*{16mm}\textbf{end}\\
        \hspace*{12mm}\textbf{end}\\
        \hspace*{8mm}\textbf{end}\\
    \end{tabular}
    \caption{Algorithm 2: the replacing part of Algorithm 1}
    \label{F:Partition Muller game strengthed}
\end{figure}

Now by repeating the proof of Lemma \ref{L: P Q}, 
we get the following.

\begin{lemma}\label{L: P Q strengthened}
    At the end of Algorithm 2, we have the following two equalities:
    \begin{itemize}
        \item  $P=\{S_i\mid i\in [1,2^n-1] \text{ and Player 0 fully wins }\mathcal{G}(S_i)\}$, and 
        \item $Q=\{S_i\mid i\in [1,2^n-1]\text{ and Player 0 cannot fully win }\mathcal{G}(S_i)\}$.
    \end{itemize}    
\end{lemma}

In the following lemmas,  we apply a binary tree to enumerate the sets $X\subseteq 2^{V}$. For each vertex $v\in V$, we maintain the number of outgoing edges from $v$ to vertices in $X$ by $out_X(v)$. During the traversing on the binary tree, there are $O(2^{|V|})$ insertions and deletions of vertices. Therefore, maintaining $out_X(v)$ takes $O(2^{|V|}\cdot |V|)$ time in total. Also let $out(v)=|E(v)|$ for $v\in V$. Then we have the following: $\mathcal{A}(X)$ is an arena if and only if for all $v\in X$, $out_X(v)\ne 0$. $\mathcal{A}(X)$ is a $\sigma$-trap  if and only if for all $v\in X\cap V_\sigma$, $out_X(v)=out(v)$ and for all $v\in X\cap V_{\bar \sigma}$, $out_X(v)>0$. For a $\bar{\sigma}$-trap $\mathcal{A}(X)$, $Attr_\sigma(X,\mathcal{A})=X$ if and only if for all $v\in V_\sigma \setminus X$, $out_X(v)=0$ and for all $v\in V_{\bar\sigma} \setminus X$, $out_X(v)<out(v)$. Hence, the following lemma is proved.

\begin{lemma}\label{L: all enumeration}
    There exists an $O(2^{|V|}\cdot |V|)$-time algorithm for each of the following enumerations:
    \begin{itemize}
        \item Enumerating all arenas $\mathcal{A}(X)$ in $\mathcal{G}$.
        \item Enumerating all $\sigma$-traps $\mathcal{A}(X)$ in $\mathcal{G}$.
        \item Enumerating all $\bar{\sigma}$-traps $\mathcal{A}(X)$ in $\mathcal{G}$ so that $Attr_\sigma(X, \mathcal{A})=X$.
    \end{itemize}
\end{lemma}

Similar to the proofs of Lemmas \ref{L: P and Q implement} and \ref{L: P to winning region implement}, applying Algorithm 2 and Lemma \ref{L: all enumeration}, we get the following lemma. 

\begin{lemma}\label{L: P and Q implement strengthen}
    There exists an algorithm that computes $P$ and $Q$ for a M\"uller game in time $O(3^{|V|}\cdot |V|)$.
\end{lemma}

\begin{lemma}\label{L: P to winning region implement strengthen}
  There is an algorithm that,  
    given M\"uller game $\mathcal{G}$ and $P$, computes $Win_0(\mathcal{G})$ and $Win_1(\mathcal{G})$ in $O(2^{|V|}\cdot |V|)$.
\end{lemma}

By Lemmas \ref{L: P and Q implement strengthen} and \ref{L: P to winning region implement strengthen}, the following theorem is proved.

\begin{theorem}\label{Thm: solving muller game strengthen}
There exists an algorithm that solves    the M\"uller game $\mathcal{G}$ in time $\mathbf{O}(3^{|V|}\cdot |V|)$.
\end{theorem}



\section{Deciding explicitly given M\"uller games in polynomial time}

This is a brief section where we describe our polynomial time algorithm that decides explicitly given M\"uller games. Currently, it is the best algorithm in terms of running times of algorithms that solve M\"uller games. For instance, when the input is exponential in the size of the arena, our algorithm outperforms with running time $O(3^{|V|}\cdot |V|)$ rather than the best known running time $O( |\Omega|\cdot(|V|+|\Omega|)\cdot |V_0|\log |V_0|  )$ from \cite{liang_et_al:LIPIcs.ESA.2023.79}. Here is the algorithm. 

On input $\mathcal G$ M\"uller game, run the following two algorithms in parallel: 
\begin{itemize}
\item Run any of our algorithms, say Algorithm 1, on $\mathcal G$, and
\item Run the polynomial time algorithm from \cite{liang_et_al:LIPIcs.ESA.2023.79} on $\mathcal G$.
\end{itemize}
Stop, once any of these algorithms outputs $W_0$ and $W_1$. 

\section{Applications}\label{S:Applications}

In this section we explain how our methods for deciding M\"uller games can be extended to all other regular games. To do so, we recast all our results in Sections \ref{S:SMG} and \ref{S:Implementations}  with an eye towards the rest of the regular games.

Lemma \ref{L: exists cannot win, exists opponent wins} and Corollary \ref{C: all cannot win, opponent wins all} and their proofs stay unchanged for all regular games.

In Lemma \ref{L: 0 no 1-trap}, Corollary \ref{C:1 X not in Omega}, Lemma \ref{L: 0 X notin Omega} and Lemma \ref{L:0 X notin Omega strengthened},  the assumption 
``$X\not \in \Omega$" is changed to the following:

\begin{itemize}
\item For McNaughton games: ``$X\cap W \not \in \Omega$'',
\item For coloured M\"uller games: ``$c(X) \not \in \Omega$'',

\item For KL games: ``For $i\in \{1,\ldots,t\}$ we have if $u_i\in X$ then $X\not\subseteq S_i$''.

\item For Rabin games: ``For $i\in \{1, \ldots, k\}$ we have if $X\cap U_i\neq \emptyset$ then $X\cap V_i\neq \emptyset$''. 

\item For Streett games: ``There is an $i\in\{1, \ldots, k\} $ such that $X\cap U_i \neq \emptyset$ and $X\cap V_i=\emptyset$''. 
\end{itemize}
Then the proofs of all the lemmas and the corollary with these new assumptions are carried out verbatim for each of these cases. Note that all requirements put on $X$ are transformations of  the winning conditions to M\"uller game winning conditions stated for Player 1. 
Similarly, in Lemma \ref{L:0 X in Omega}
the assumption 
``$X\in \Omega$" is changed to the following:

\begin{itemize}
\item For McNaughton games: ``$X\cap W \in \Omega$'',
\item For coloured M\"uller games: ``$c(X) \in \Omega$'',

\item For KL games: ``There is an $i\in\{1, \ldots, t\} $ such that $u_i\in X$ and $X\subseteq S_i$''.

\item For Rabin games: ``There is an $i\in\{1, \ldots, k\} $ such that $X\cap U_i \neq \emptyset$ and $X\cap V_i=\emptyset$''.
\item For Streett games: ``For $i\in \{1, \ldots, k\}$ we have if $X\cap U_i\neq \emptyset$ then $X\cap V_i\neq \emptyset$''. 
\end{itemize}
Then the proof of the lemma with these new assumptions is carried out word by word for each of the cases. Just as above, all the conditions put in $X$ are essentially transformation of the games to M\"uller games stated for Player 0. 

It is not too hard to see that for McNaughton games and coloured M\"uller games, we can easily recast the algorithms presented in Section \ref{S:Implementations}. There will be no influence on running time complexity. Hence, we get the following complexity-theoretic result as in Theorem \ref{Thm: solving muller game strengthen}:

\begin{theorem}
There exist algorithms that decide McNaughton and coloured M\"uller games $\mathcal{G}$ in time $\mathbf{O}(3^{|V|}\cdot |V|)$. \qed
\end{theorem}

Note that the algorithms presented in Section \ref{S:Implementations} can also be applied to KL, Rabin and Streett games. However, one needs to be careful with the parameters involved. They add additional running time costs.
Namely, the algorithms should verify the assumptions, put on the sets $X$, dictated by KL, Rabin and Streett conditions.  


We start with the transformation from KL games to M\"uller games. 
Let $\mathcal{G}=(\mathcal{A},(u_1, S_1), \ldots, (u_t, S_t))$ be a KL game. Define the following M\"uller condition
set $\Omega'$: \ 
$X\in \Omega' \ \mbox{if and only if for some pair $(u_i, S_i)$ we have} \  u_i\in X\mbox{ and } X\subseteq S_i$. 

\begin{lemma}
    The transformation from KL games to M\"uller games takes $O(3^{|V|}|V|^2)$ time.
\end{lemma}
\begin{proof}
     We apply a binary tree to maintain $\Omega'$. Then enumerate all pairs $(u_i, S_i)$ and add all $X$ with $u_i\in X$ and $X\subseteq S_i$ into $\Omega'$. Let $\mathcal{S}_i$ be the set of all $S_{i,j}\subseteq V$ so that $(v_i, S_{i,j})$ is a winning condition. Since $\mathcal{S}_i\subseteq 2^{V}$, for all pairs $(u_k, S_k)$ with $u_k=v_i$, there are at most  $3^{|V|}$ additions of $X$s. Therefore, the transformation takes $O(3^{|V|}|V|^2)$ time in total. 
\end{proof}
As an immediate corollary we get the following complexity-theoretic result for KL games. 
\begin{theorem} \label{C:KL-M}
 There exists an algorithm that, given a KL game $\mathcal G$, decides $\mathcal G$ in time 
 $O(3^{|V|}|V|^2)$. \qed 
\end{theorem}

Now we transform Rabin games $\mathcal G$ to M\"uller games. Direct translation to M\"uller games is costly as each pair $(U_i,V_i)$ in the Rabin winning condition defines the collection of sets $X$ such that $X\cap U_i\neq \emptyset$ and $X\cap V_i=\emptyset$. The collection of all these sets $X$ form the M\"uller condition set $\Omega$. 
As the index $k$ is $O(2^{2|V|})$, the direct transformation is expensive. Our goal is to avoid this cost through $KL$ games.  The following lemma is easy:

\begin{lemma}\label{L: Ui Vi to Yi Zi}
     Let $X\subseteq V$ and let $(U_i, V_i)$ be a winning pair in Rabin game $\mathcal G$.
     Set $Y_i=U_i\setminus V_i$ and $Z_i=V\setminus V_i$. Then $X\cap U_i\ne\emptyset$ and $X\cap V_i=\emptyset$ if and only if $X\cap Y_i\ne\emptyset$ and $X\subseteq Z_i$.
\end{lemma}

Thus, we can replace the winning condition $(U_1, V_1), \ldots (U_k, V_k)$ in a given Rabin game  to the equivalent winning condition $(Y_1, Z_1), \ldots, (Y_k, Z_k)$. We still have Rabin winning condition but we use this 
new winning condition $(Y_1, Z_1), \ldots, (Y_k, Z_k)$ to build the desired $KL$ game:  

\begin{lemma}
    The transformation from Rabin games to KL games takes $O(k|V|^2)$ time.
\end{lemma}
\begin{proof}
    Enumerate all pairs $(U_i, V_i)$, compute $Y_i=U_i\setminus V_i$, $Z_i=V\setminus V_i$ and add all pairs $(u_j, S_j)$ with $u_j\in Y_i$ and $S_j=Z_i$ into KL conditions. By applying binary trees, the transformation takes $O(k|V|^2)$. This transformation preserves the winning sets $W_0$ and $W_1$. 
\end{proof}

Thus, the transformed KL games can be viewed as a compressed version of Rabin games. 

\begin{corollary}
    The transformation from Rabin games to M\"uller games takes $O((k+3^{|V|})|V|^2)$ time.
\end{corollary}

Note that deciding Rabin games is equivalent to deciding Streett games. Thus, combining the arguments above, we get the following complexity-theoretic result:
\begin{theorem}
 There exist algorithms that decide Rabin and Streett games $\mathcal{G}$ in time $\mathbf{O}((k+3^{|V|})\cdot |V|^2)$. \qed
\end{theorem}

\section{Conclusion}

The algorithms presented in this work give rise to numerous questions that warrant further exploration. For instance, we know that explicitly given M\"uller games can be decided in polynomial time. Yet, we do not know if there are polynomial time algorithms that decide explicitly given  McNaughton games and coloured M\"uller games. Another intriguing line of research is to establish connections between our algorithms and the parameters of the games, with the aim of incorporating these parameters into the running time analysis. Another natural question is to try to decrease the base $3$ in the running times of our algorithms, thereby further optimizing the efficiency. This reduction of computational overhead may uncover new insights and lead to even more efficient algorithms.

\bibliographystyle{ACM-Reference-Format}
\bibliography{bibfile}


\begin{thebibliography}{14}


\ifx \showCODEN    \undefined \def \showCODEN     #1{\unskip}     \fi
\ifx \showDOI      \undefined \def \showDOI       #1{#1}\fi
\ifx \showISBNx    \undefined \def \showISBNx     #1{\unskip}     \fi
\ifx \showISBNxiii \undefined \def \showISBNxiii  #1{\unskip}     \fi
\ifx \showISSN     \undefined \def \showISSN      #1{\unskip}     \fi
\ifx \showLCCN     \undefined \def \showLCCN      #1{\unskip}     \fi
\ifx \shownote     \undefined \def \shownote      #1{#1}          \fi
\ifx \showarticletitle \undefined \def \showarticletitle #1{#1}   \fi
\ifx \showURL      \undefined \def \showURL       {\relax}        \fi
\providecommand\bibfield[2]{#2}
\providecommand\bibinfo[2]{#2}
\providecommand\natexlab[1]{#1}
\providecommand\showeprint[2][]{arXiv:#2}

\bibitem[Bj{\"o}rklund et~al\mbox{.}(2003)]%
        {bjorklund2003fixed}
\bibfield{author}{\bibinfo{person}{Henrik Bj{\"o}rklund}, \bibinfo{person}{Sven
  Sandberg}, {and} \bibinfo{person}{Sergei Vorobyov}.}
  \bibinfo{year}{2003}\natexlab{}.
\newblock \showarticletitle{On fixed-parameter complexity of infinite games}.
  In \bibinfo{booktitle}{\emph{The Nordic Workshop on Programming Theory (NWPT
  2003)}}, Vol.~\bibinfo{volume}{34}. Citeseer, \bibinfo{pages}{29--31}.
\newblock


\bibitem[Calude et~al\mbox{.}(2017)]%
        {calude2017deciding}
\bibfield{author}{\bibinfo{person}{Cristian~S Calude}, \bibinfo{person}{Sanjay
  Jain}, \bibinfo{person}{Bakhadyr Khoussainov}, \bibinfo{person}{Wei Li},
  {and} \bibinfo{person}{Frank Stephan}.} \bibinfo{year}{2017}\natexlab{}.
\newblock \showarticletitle{Deciding parity games in quasipolynomial time}. In
  \bibinfo{booktitle}{\emph{Proceedings of the 49th Annual ACM SIGACT Symposium
  on Theory of Computing}}. \bibinfo{pages}{252--263}.
\newblock
\newblock
\shownote{STOC 2017 Best Paper Award}.


\bibitem[Dziembowski et~al\mbox{.}(1997)]%
        {dziembowski1997much}
\bibfield{author}{\bibinfo{person}{Stefan Dziembowski}, \bibinfo{person}{Marcin
  Jurdzinski}, {and} \bibinfo{person}{Igor Walukiewicz}.}
  \bibinfo{year}{1997}\natexlab{}.
\newblock \showarticletitle{How much memory is needed to win infinite games?}.
  In \bibinfo{booktitle}{\emph{Proceedings of Twelfth Annual IEEE Symposium on
  Logic in Computer Science}}. IEEE, \bibinfo{pages}{99--110}.
\newblock


\bibitem[Fijalkow et~al\mbox{.}(2023)]%
        {fijalkow2023games}
\bibfield{author}{\bibinfo{person}{Nathanaël Fijalkow},
  \bibinfo{person}{Nathalie Bertrand}, \bibinfo{person}{Patricia
  Bouyer-Decitre}, \bibinfo{person}{Romain Brenguier}, \bibinfo{person}{Arnaud
  Carayol}, \bibinfo{person}{John Fearnley}, \bibinfo{person}{Hugo Gimbert},
  \bibinfo{person}{Florian Horn}, \bibinfo{person}{Rasmus Ibsen-Jensen},
  \bibinfo{person}{Nicolas Markey}, \bibinfo{person}{Benjamin Monmege},
  \bibinfo{person}{Petr Novotný}, \bibinfo{person}{Mickael Randour},
  \bibinfo{person}{Ocan Sankur}, \bibinfo{person}{Sylvain Schmitz},
  \bibinfo{person}{Olivier Serre}, {and} \bibinfo{person}{Mateusz Skomra}.}
  \bibinfo{year}{2023}\natexlab{}.
\newblock \bibinfo{title}{Games on Graphs}.
\newblock
\newblock
\showeprint[arxiv]{2305.10546}~[cs.GT]
\newblock
\shownote{To be published by Cambridge University Press. Editor: Nathanaël
  Fijalkow}.


\bibitem[Gr{\"a}del et~al\mbox{.}(2002)]%
        {gradel2002automata}
\bibfield{author}{\bibinfo{person}{Erich Gr{\"a}del}, \bibinfo{person}{Wolfgang
  Thomas}, {and} \bibinfo{person}{Thomas Wilke}.}
  \bibinfo{year}{2002}\natexlab{}.
\newblock \bibinfo{title}{Automata, logics, and infinite Games. LNCS, vol.
  2500}.
\newblock
\newblock


\bibitem[Hoorfar and Hassani(2008)]%
        {hoorfar2008inequalities}
\bibfield{author}{\bibinfo{person}{Abdolhossein Hoorfar} {and}
  \bibinfo{person}{Mehdi Hassani}.} \bibinfo{year}{2008}\natexlab{}.
\newblock \showarticletitle{Inequalities on the Lambert W function and
  hyperpower function}.
\newblock \bibinfo{journal}{\emph{J. Inequal. Pure and Appl. Math}}
  \bibinfo{volume}{9}, \bibinfo{number}{2} (\bibinfo{year}{2008}),
  \bibinfo{pages}{5--9}.
\newblock


\bibitem[Horn(2008)]%
        {horn2008explicit}
\bibfield{author}{\bibinfo{person}{Florian Horn}.}
  \bibinfo{year}{2008}\natexlab{}.
\newblock \showarticletitle{Explicit Muller games are PTIME}. In
  \bibinfo{booktitle}{\emph{IARCS Annual Conference on Foundations of Software
  Technology and Theoretical Computer Science}}. Schloss
  Dagstuhl-Leibniz-Zentrum f{\"u}r Informatik.
\newblock


\bibitem[Hunter and Dawar(2008)]%
        {hunter2008complexity}
\bibfield{author}{\bibinfo{person}{Paul Hunter} {and} \bibinfo{person}{Anuj
  Dawar}.} \bibinfo{year}{2008}\natexlab{}.
\newblock \showarticletitle{Complexity bounds for muller games}.
\newblock \bibinfo{journal}{\emph{Theoretical Computer Science (TCS)}}
  (\bibinfo{year}{2008}).
\newblock


\bibitem[Liang et~al\mbox{.}(2023)]%
        {liang_et_al:LIPIcs.ESA.2023.79}
\bibfield{author}{\bibinfo{person}{Zihui Liang}, \bibinfo{person}{Bakh
  Khoussainov}, \bibinfo{person}{Toru Takisaka}, {and} \bibinfo{person}{Mingyu
  Xiao}.} \bibinfo{year}{2023}\natexlab{}.
\newblock \showarticletitle{{Connectivity in the Presence of an Opponent}}. In
  \bibinfo{booktitle}{\emph{31st Annual European Symposium on Algorithms (ESA
  2023)}} \emph{(\bibinfo{series}{Leibniz International Proceedings in
  Informatics (LIPIcs)}, Vol.~\bibinfo{volume}{274})},
  \bibfield{editor}{\bibinfo{person}{Inge~Li G{\o}rtz}, \bibinfo{person}{Martin
  Farach-Colton}, \bibinfo{person}{Simon~J. Puglisi}, {and}
  \bibinfo{person}{Grzegorz Herman}} (Eds.). \bibinfo{publisher}{Schloss
  Dagstuhl -- Leibniz-Zentrum f{\"u}r Informatik}, \bibinfo{address}{Dagstuhl,
  Germany}, \bibinfo{pages}{79:1--79:14}.
\newblock
\showISBNx{978-3-95977-295-2}
\showISSN{1868-8969}
\urldef\tempurl%
\url{https://doi.org/10.4230/LIPIcs.ESA.2023.79}
\showDOI{\tempurl}


\bibitem[McNaughton(1993)]%
        {mcnaughton1993infinite}
\bibfield{author}{\bibinfo{person}{Robert McNaughton}.}
  \bibinfo{year}{1993}\natexlab{}.
\newblock \showarticletitle{Infinite games played on finite graphs}.
\newblock \bibinfo{journal}{\emph{Annals of Pure and Applied Logic}}
  \bibinfo{volume}{65}, \bibinfo{number}{2} (\bibinfo{year}{1993}),
  \bibinfo{pages}{149--184}.
\newblock


\bibitem[Neider et~al\mbox{.}(2014)]%
        {neider2014down}
\bibfield{author}{\bibinfo{person}{Daniel Neider}, \bibinfo{person}{Roman
  Rabinovich}, {and} \bibinfo{person}{Martin Zimmermann}.}
  \bibinfo{year}{2014}\natexlab{}.
\newblock \showarticletitle{Down the Borel hierarchy: Solving Muller games via
  safety games}.
\newblock \bibinfo{journal}{\emph{Theoretical Computer Science}}
  \bibinfo{volume}{560} (\bibinfo{year}{2014}), \bibinfo{pages}{219--234}.
\newblock


\bibitem[Nerode et~al\mbox{.}(1996)]%
        {nerode1996mcnaughton}
\bibfield{author}{\bibinfo{person}{Anil Nerode}, \bibinfo{person}{Jeffrey~B
  Remmel}, {and} \bibinfo{person}{Alexander Yakhnis}.}
  \bibinfo{year}{1996}\natexlab{}.
\newblock \showarticletitle{McNaughton games and extracting strategies for
  concurrent programs}.
\newblock \bibinfo{journal}{\emph{Annals of Pure and Applied Logic}}
  \bibinfo{volume}{78}, \bibinfo{number}{1-3} (\bibinfo{year}{1996}),
  \bibinfo{pages}{203--242}.
\newblock


\bibitem[Piterman and Pnueli(2006)]%
        {piterman2006faster}
\bibfield{author}{\bibinfo{person}{Nir Piterman} {and} \bibinfo{person}{Amir
  Pnueli}.} \bibinfo{year}{2006}\natexlab{}.
\newblock \showarticletitle{Faster solutions of Rabin and Streett games}. In
  \bibinfo{booktitle}{\emph{21st Annual IEEE Symposium on Logic in Computer
  Science (LICS'06)}}. IEEE, \bibinfo{pages}{275--284}.
\newblock


\bibitem[Zielonka(1998)]%
        {zielonka1998infinite}
\bibfield{author}{\bibinfo{person}{Wieslaw Zielonka}.}
  \bibinfo{year}{1998}\natexlab{}.
\newblock \showarticletitle{Infinite games on finitely coloured graphs with
  applications to automata on infinite trees}.
\newblock \bibinfo{journal}{\emph{Theoretical Computer Science}}
  \bibinfo{volume}{200}, \bibinfo{number}{1-2} (\bibinfo{year}{1998}),
  \bibinfo{pages}{135--183}.
\newblock


\end{thebibliography}

\end{document}